\documentclass{IEEEtran}
\usepackage[T1]{fontenc}
\usepackage{array}
\usepackage{algorithm}
\usepackage[noend]{algorithmic}
\usepackage{amsmath}
\usepackage{amssymb}
\usepackage{amsthm}
\usepackage{caption}
\usepackage{colortbl}
\usepackage{diagbox}
\usepackage[super]{nth}
\usepackage{pgfplots}
\usepackage{tikz}

\usetikzlibrary{calc}
\usetikzlibrary{positioning}
\usetikzlibrary{shapes.geometric}

\tikzstyle{node} = [circle, draw, minimum size = 0.6 cm]

\newtheorem{theorem}{Theorem}
\newtheorem{assumption}[theorem]{Assumption}
\newtheorem{condition}[theorem]{Condition}
\newtheorem{corollary}[theorem]{Corollary}
\newtheorem{definition}[theorem]{Definition}
\newtheorem{invariant}[theorem]{Invariant}
\newtheorem{lemma}[theorem]{Lemma}
\newtheorem{proposition}[theorem]{Proposition}

\newcolumntype{x}[1]{>{\centering\arraybackslash\hspace{0pt}}p{#1}}

\DeclareMathOperator{\AU}{AU}
\DeclareMathOperator{\cost}{cost}
\DeclareMathOperator{\edge}{edge}
\newcommand{\lbl}[2]{(#1, #2)}
\DeclareMathOperator{\n}{n}
\DeclareMathOperator{\I}{I}
\DeclareMathOperator{\myrelax}{relax}
\DeclareMathOperator{\source}{source}
\DeclareMathOperator{\target}{target}
\DeclareMathOperator{\RI}{RI}

\newcommand{\eattrs}[2]{(#1, #2)}

\newcommand{\RIO}[2]{[#1, #2)}

\newcommand{\RomanI}{I}
\newcommand{\RomanII}{{I\hspace{-0.05cm}I}}
\newcommand{\RomanIII}{{I\hspace{-0.05cm}I\hspace{-0.05cm}I}}

\title{Generic Dijkstra}

\author{Ireneusz Szcześniak, Bożena Woźna-Szcześniak, Ireneusz
  Olszewski

  \thanks{I.~Szcześniak is with the Department of Computer Science of
    the Częstochowa University of Technology, Poland}

  \thanks{B.~Woźna-Szcześniak is with the Department of Mathematics
    and Computer Science of the Jan Długosz University in Częstochowa,
    Poland}

  \thanks{I.~Olszewski is with the Faculty of Telecommunications,
    Computer Science and Electrical Engineering of the Bydgoszcz
    University of Science and Technology, Poland} }

\begin{document}

\maketitle

\begin{abstract}
  The recently-proposed generic Dijkstra algorithm finds shortest
  paths in networks with continuous and contiguous resources.  While
  the algorithm was proposed in the context of optical networks (and
  is applicable to other networks with finite and discrete resources),
  we present the stated problem in a broader algorithmic setting of
  the greedy approach.  The algorithm was published without a proof of
  correctness, and with a minor shortcoming.  We provide that missing
  proof and offer a correction to the shortcoming.  To prove the
  algorithm correct, we generalize the greedy approach and the Bellman
  principle of optimality to algebraic structures with a partial
  ordering.  By analyzing the size of the search space in the
  worst-case, we argue the stated problem is tractable.  Thus we
  definitely answer a long-standing fundamental question of whether we
  can efficiently find a shortest path in a network with discrete
  resources subject to the continuity and contiguity constraints: yes,
  we can.
\end{abstract}


\section{Introduction}

% What problem the generic Dijkstra algorithm solves.

The Dijkstra shortest-path algorithm finds shortest paths in a graph
from a given source vertex to all other vertexes, provided the costs
of edges are non-negative \cite{10.1007/BF01386390}.  However, the
algorithm cannot be used if the paths found have to meet the resource
\emph{continuity} and \emph{contiguity} constraints.

% Finite and discrete resource.

Some networks have \emph{finite} and \emph{discrete} resources which
have to be used by a connection along the path.  For instance, in
communication networks, frequency or time is such resource.  The
resource is divided into discrete units modeled by a set of integers
$\Omega = [0, U)$, where $U$ is a \emph{finite} number of units
offered.  Each edge has a set of available units, which is a subset
of $\Omega$.

% What's the continuity constraint.

A path meets the resource continuity constraint if some given units
are available on each of its edges.  Therefore an algorithm should
find not only a path, but also the set of units available along the
path.  When a connection is established along the found path, the
required units are made unavailable along the path (on its edges),
thus changing the state of the network, and affecting the future path
searches.

% Why the continuity constraint.

The continuity constraint follows from the characteristics of the
modeled network.  For instance, in optical networks, a unit is a
wavelength or a frequency unit, which has to be used along the path in
order to take advantage of the wavelength-division multiplexing.  In
networks capable of advance resource reservation, a unit is a time
slot, which has to be used along the path in order to schedule
transmission at a given time along the edges used.

% What's the contiguity constraint.

The resource contiguity constraint requires that the units of a
connection be contiguous.  We model contiguous units by a half-closed
integer interval.  For example, the units of interval $[0, 3)$ are
contiguous, and the units of set $\{0, 2\}$ are not.

% Why the contiguity constraint.  
  
The contiguity constraint also follows from the characteristics of the
modeled network.  For instance, in the elastic optical networks, the
resource contiguity constraint follows from the physical and
economical limitations of optical networks, where the frequency units
of a connection should occupy a single frequency band.  In networks
capable of advance resource reservation, the (time slot) contiguity
constraint follows from the assumption that a connection should last
unintermitted.

% What the generic Dijkstra algorithm is.

The generic Dijkstra algorithm \cite{10.1364/JOCN.11.000568} finds
efficiently shortest paths that meet the resource continuity and
contiguity constraints.  Both the standard and generic Dijkstra
algorithms are \emph{comparison algorithms}, i.e., they compare labels
based on their ordering.  The generic Dijkstra algorithm is a
generalization of the standard Dijkstra algorithm: while the standard
Dijkstra requires the ordering (of the cost labels) to be total, the
generic Dijkstra relaxes this requirement and allows for a partial
ordering.

% Contribution of the paper.

Our novel contribution is the generalization of the greedy approach
and the Bellman principle of optimality to partial orderings, the
proof of correctness of the generic Dijkstra algorithm, and a
correction to a minor shortcoming of the algorithm.  By proposing the
worst-case analysis of the search space, we argue that the stated
problem is tractable, a crucial conclusion not only from a theoretical
perspective, but from a practical one as well: for network operations
and management with the ever increasing performanance requirements.

% Article organization.

\section{Related works}

% The extension of the conference paper.

This article is an extension of
\cite{10.1109/NOMS56928.2023.10154322}.  We elaborated on the
preliminaries and the terminology (Section \ref{s:preliminaries}) to
better explain the label incomparability and its intransitivity that
set it apart from multiobjective optimization.  We introduced a new
section (Section \ref{s:greedy}) to generalize the greedy approach.
Also, we improved Section \ref{s:bellman}, and various other parts of
the article.

% On the algorithm.

In \cite{10.1364/JOCN.11.000568}, the generic Dijkstra algorithm was
proposed to solve the dynamic routing and spectrum assignment (RSA)
problem in optical networks.  There a path was discarded if it was
unable to support a demand due to modulation constraints.  Here,
however, we concentrate on proving the algorithm correctness in the
most general setting, without considering the constraints of specific
networks.  The algorithm was simulatively demonstrated to efficiently
find exact results by others \cite{9464058}.  To the best of our
knowledge, no proof of algorithm correctness has been published.  The
algorithm was adapted to solve the dynamic routing problem with
dedicated path protection \cite{10.3390/e23091116} but without a proof
of correctness.

% The principle of optimality.

Bellman first formulated the \emph{principle of optimality} in
\cite{10.1073/pnas.38.8.716}, and then applied it to the shortest-path
problem in \cite{bellman}.  This principle is not exactly a solution
to a problem (which Bellman readily acknowledged in \cite{bellman}
that the principle did not offer a ready computational scheme), but
rather a stipulation the solution should meet \cite{sniedovich}.  Thus
the principle defines a class of problems whose solutions have the
optimal substructure, i.e., whose subproblem has an optimal solution.
For the shortest-path problem, the principle stipulates that shortest
paths are made of shortest paths, and form a shortest-path tree
\cite{isbn/0262031418}.  Since dynamic programming relies on the
optimal substructure, this principle is also known as the
\emph{dynamic programming principle}.

% The Infocom 1993 paper.

In \cite{10.1109/INFCOM.1993.253315} the authors proposed an algorithm
for finding a shortest path in a network with continuous and discrete
resources.  While we believe the algorithm is exact, the authors call
is heuristic, albeit in the context of optimizing the overall network
performance.  To the best of our knowledge, no proof of correctness
was proposed for that algorithm.

% The status was unclear.

To the best of our knowledge, the status of this routing problem had
been unclear: no proof of nondeterministic-polynomial (NP)
completeness was proposed, no efficient algorithm with proven
correctness has been published but heuristic algorithms were commonly
used, most notably based on a K-shortest path algorithm or on the
Dijkstra algorithm \cite{10.1007/s11107-017-0700-5}.

% Why the problem is tractable. The filtered-graphs algorithm.

However, searching for a path that meets the resource continuity and
contiguity constraints is an easy (in terms of complexity theory)
problem, because the number of shortest paths is polynomially bounded
\cite{10.1364/JOCN.11.000568}.  The problem can be solved in
worst-case polynomial time by searching for a shortest path (with the
standard Dijkstra algorithm) in a polynomially-bounded number of
filtered-graphs, and choosing the shortest from the shortest paths
found.  A filtered graph retains only those edges from the input graph
that have specific contiguous units available.

% Proofs of the Dijkstra algorithm.

The proofs of correctness of the Dijkstra algorithm have been offered
in a number of mile-stone text books \cite{isbn/0201000296,
isbn/0262031418, networkflows}.  Usually the proofs are inductive
and combined with contradiction.  We find the most convincing the
proof offered in \cite{networkflows} which does not rely on
contradiction, and that strategy we use in the proof we propose.

% The Martins algorithm.

The Martins algorithm is a general multilabeling (i.e., allowing
multiple labels for a vertex) algorithm for solving optimization
problems with multiple totally ordered objectives
\cite{10.1016/0377-2217(84)90077-8}.  The algorithm could be used to
solve the stated problem but with the exponential worst-case
complexity.  The generic Dijkstra algorithm is similar to the Martins
algorithm in that it is also multilabeling, but different in that it
is of a single partially ordered objective.

% Dominance relation: dominated vs nondominated.

In multiobjective optimization, two solutions can be related by the
\emph{dominance} relation.  If the relation holds, we say that one
solution dominates the other (a dominated solution), otherwise
solutions are nondominated.  The dominance relation is defined for
totally ordered objectives: one solution dominates the other solution
if one objective is smaller for one solution, while the other
objectives are smaller or equal \cite{10.1016/0377-2217(84)90077-8}
for the other solution -- that definition has been used ever since,
e.g., \cite{isbn/3540213988}, but in
\cite{10.1007/978-3-642-48782-8_9} a slightly different, though
equivalent, definition was first proposed for two objectives.

% Nondominance vs incomparability.

The nondominance relation holds when one solution does not dominate
the other (the dominance relation is partial).  The only reason why
the nondominance relation holds is that one objective is smaller for
one solution, while some other objective is smaller for the other
solution.  However, the incomparability relation can also hold (and
that is the difference between nondominance and incomparability),
because resources are incomparable.

% ********************************************************************
% PRELIMINARIES
% ********************************************************************

\section{Preliminaries and terminology}
\label{s:preliminaries}

\subsection{Resource interval}
\label{s:ri}

% What an RI is.

We refer to a half-closed integer interval as a resource interval
(RI).  We describe a set of contiguous units with an RI, and refer to
it with $r$.  For instance, by $r = \RIO{10}{12}$ we mean units 10 and
11.  Function $\min(r)$ returns the lower endpoint of $r$, and
$\max(r)$ the upper endpoint.  For example, $\min(r) = 10$, and
$\max(r) = 12$.

\subsubsection{Relation $\supset$}

% The inclusion relation: proper or not.  The converse.  Transitive.

We look for the largest RIs using the inclusion relation: the larger,
the better, as it describes also the included RIs.  For instance, $r_1
= \RIO{0}{1}$ is worse than $r_3 = \RIO{0}{2}$ (or $r_3$ is better
than $r_1$), because $r_3$ properly includes $r_1$, i.e., $r_3 \supset
r_1$.  Relation $\supset$ is given by (\ref{e:r_supset}) that holds
for the three top-right white cells of Table \ref{t:r_rels}, while the
converse ($\subset$) holds for the three bottom-left white cells.
When we add equality to $\supset$, we get $\supseteq$ as given by
(\ref{e:r_supseteq}) that also holds for the black cell.  A computer
executes $\supseteq$ faster than $\supset$.  These relations are
transitive.

\begin{equation}
  \begin{split}
    r_i \supset r_j \iff & \min(r_i) < \min(r_j) \text{ and}\\
    & \max(r_i) \ge \max(r_j) \text{ or}\\
    & \min(r_i) \le \min(r_j) \text{ and}\\
    & \max(r_i) > \max(r_j)
  \end{split}
  \label{e:r_supset}
\end{equation}

\begin{equation}
  \begin{split}
    r_i \supseteq r_j \iff & \min(r_i) \le \min(r_j) \text{ and}\\
    & \max(r_i) \ge \max(r_j)
  \end{split}
  \label{e:r_supseteq}
\end{equation}

% Comparability, partial ordering, incomparability, symmetric
% complement, and intransitivity.

RIs are $\supseteq$-comparable (or comparable in short), if one
includes the other, properly or not, i.e., $\supseteq$ or $\subseteq$
holds.  However, there can exist RIs for which neither $\supseteq$ nor
$\subseteq$ holds, and so $\supseteq$ is a \emph{partial ordering} (as
it is also reflexive and antisymmetric).  RIs are
$\supseteq$-incomparable (or incomparable in short), denoted by
$\parallel$, if one does not include the other, properly or not, i.e.,
$\not\supseteq$ and $\not\subseteq$ hold (the shaded cells in Table
\ref{t:r_rels}), as given by (\ref{e:r_incomparability}).  For
instance, $r_1$ is incomparable with $r_2 = \RIO{1}{3}$, i.e., $r_1
\parallel r_2$, because neither $\supseteq$ nor $\subseteq$ holds.
Relation $\parallel$ is the symmetric complement of $\supseteq$, and
is intransitive: e.g., $r_1 \parallel r_2$ and $r_2 \parallel r_3$
does not imply $r_1 \parallel r_3$ because $r_1 \subset r_3$.
Therefore, $\parallel$ is not an equivalence relation; incomparable
RIs are not equivalent.

\begin{equation}
  r_i \parallel r_j \iff r_i \not\supseteq r_j \text{ and }
  r_i \not\subseteq r_j
  \label{e:r_incomparability}
\end{equation}

Even though incomparability of RIs is not transitive, the tractability
of the stated problem depends on Proposition \ref{p:savetheday1}
related to the incomparability.

\begin{proposition}
  $r_i \parallel r_j \supseteq r_k \implies r_i \nsubseteq r_k$
  \label{p:savetheday1}
\end{proposition}

\begin{proof}
  Relation $r_i \parallel r_j$ holds in either of two cases: first,
  $\min(r_i) < \min(r_j)$ and $\max(r_i) < \max(r_j)$; second,
  $\min(r_i) > \min(r_j)$ and $\max(r_i) > \max(r_j)$.  Furthermore,
  by (\ref{e:r_supseteq}), $r_j \supseteq r_k \iff \min(r_j) \le
  \min(r_k) \text{ and} \max(r_j) \ge \max(r_k)$.

  In the first case, $\min(r_i) < \min(r_j) \le \min(r_k) \implies
  \min(r_i) < \min(r_k) \implies r_i \nsubseteq r_k$ (the first row of
  Table \ref{t:r_rels}).  In the second case, $\max(r_i) > \max(r_j)
  \ge \max(r_k) \implies \max(r_i) > \max(r_k) \implies r_i \nsubseteq
  r_k$ (the third column of Table \ref{t:r_rels}).
\end{proof}

% The proper inclusion is a strict partial ordering, so it deserves to
% ask whether it's also a weak ordering that is strict.  With the
% symmetric complement, equivalence relation and intransitivity
% introduced, we show that the proper inclusion is not a weak
% ordering.

Relation $\supset$ is a strict partial ordering (as it is also
asymmetric), but not a \emph{weak ordering}\footnote{A weak ordering
is strict, i.e., irreflexive: Lemma 4.6 in \cite{elements}.
Nonetheless, in programming, e.g., in C++, this ordering is on purpose
superfluously called a \emph{strict weak ordering} to hint programmers
that the ordering implementation must not include equivalence.}.  The
symmetric complement of a weak ordering is the equivalence relation
whose equivalence classes are linearly ordered \cite{elements}, and
that is enough to sort by comparison.  However, the symmetric
complement of $\supset$ is not an equivalence relation (since it is
intransitive: e.g., $r_1 \not\supset r_2$ and $r_1 \not\subset r_2$,
$r_2 \not\supset r_3$ and $r_2 \not\subset r_3$, yet $r_1 \subset
r_3$), and so comparison algorithms cannot use $\supset$.

\subsubsection{Relation $>$}

% We introduce the ordering needed for sorting.

The generic Dijkstra algorithm (since it is a comparison algorithm)
requires at least a weak ordering of RIs, and to that end we introduce
a total ordering $>$, an extension of $\supset$, as defined in Table
\ref{t:r_rels}, where a cell in its top-left part reports on the
$\supset$ relation, and in the bottom-right on $>$.

\begin{table*}
  \caption{Relations between RIs $r_i$ and $r_j$.}
  \label{t:r_rels}
  \centering%
  \begin{tabular}{c|c|c|c|}
    \cline{2-4}
    & $\max(r_i) < \max(r_j)$
    & $\max(r_i) = \max(r_j)$
    & $\max(r_i) > \max(r_j)$\\
    \hline
    \multicolumn{1}{|c|}{$\min(r_i) < \min(r_j)$}
    & \cellcolor[gray]{0.8}\diagbox[dir = SW]{$r_i \parallel r_j$}{$r_i > r_j$}
    & \diagbox[dir = SW]{$r_i \supset r_j$}{$r_i > r_j$}
    & \diagbox[dir = SW]{$r_i \supset r_j$}{$r_i > r_j$}\\
    \hline
    \multicolumn{1}{|c|}{$\min(r_i) = \min(r_j)$}
    & \diagbox[dir = SW]{$r_i \subset r_j$}{$r_i < r_j$}
    & \cellcolor[gray]{0}\color{white}{$r_i = r_j$}
    & \diagbox[dir = SW]{$r_i \supset r_j$}{$r_i > r_j$}\\
    \hline
    \multicolumn{1}{|c|}{$\min(r_i) > \min(r_j)$}
    & \diagbox[dir = SW]{$r_i \subset r_j$}{$r_i < r_j$}
    & \diagbox[dir = SW]{$r_i \subset r_j$}{$r_i < r_j$}
    & \cellcolor[gray]{0.8}\diagbox[dir = SW]{$r_i \parallel r_j$}{$r_i < r_j$}\\
    \hline
  \end{tabular}
\end{table*}

% Definition of $>$.

Relation $r_i \supset r_j$ has to imply $r_i > r_j$ (the unshaded
cells in Table \ref{t:r_rels}) because $r_i$ should be processed first
(before $r_j$) as it offers a better solution (see Section
\ref{s:greedy})\footnote{We introduced the $<$ relation for RIs with
Eq.~(1) in \cite{10.1109/NOMS56928.2023.10154322} where $\supset$
implied $<$, but here we reversed its direction, so that $\supset$
implies $>$.}.  The shaded cells (in Table \ref{t:r_rels}) report on
one of two possibilities of defining $>$ for the incomparable RIs so
that $>$ is transitive.  The other possibility would have the relation
flipped in both gray cells.  The choice is arbitrary, and we prefer
the RIs of smaller lower endpoints to be processed first, i.e., $r_i >
r_j$ if $min(r_i) < min(r_j)$.  The other possibility would have the
RIs of greater upper endpoints to be processed first, i.e., $r_i >
r_j$ if $max(r_i) > max(r_j)$.  The ordering defined by (\ref{e:r_>})
is lexicographic by the lower endpoints first (using $<$), and by the
upper second (using $>$).  Relation $\ge$ is given by (\ref{e:r_>=}).

\begin{equation}
  \begin{split}
    r_i > r_j \iff & \min(r_i) < \min(r_j) \text{ or}\\
    & \min(r_i) = \min(r_j) \text{ and}\\
    & \max(r_i) > \max(r_j)
  \end{split}
  \label{e:r_>}
\end{equation}

\begin{equation}
  \begin{split}
    r_i \ge r_j \iff & \min(r_i) < \min(r_j) \text{ or}\\
    & \min(r_i) = \min(r_j) \text{ and}\\
    & \max(r_i) \ge \max(r_j)
  \end{split}
  \label{e:r_>=}
\end{equation}

\begin{lemma}
  Relation $>$ for RIs is transitive.
  \label{l:r_trans}
\end{lemma}

\begin{proof}
  Relation $>$ is transitive because $r_i > r_j > r_k$ implies $r_i >
  r_k$.  Relation $r_i > r_j$ holds in either of two cases: \nth{1},
  if $min(r_i) < min(r_j)$; and \nth{2}, if $min(r_i) = min(r_j)$ and
  $max(r_i) > max(r_j)$.  If both $r_i > r_j$ and $r_j > r_k$ hold in
  the \nth{1} case, then $min(r_i) < min(r_j) < min(r_k)$ holds, and
  so $r_i > r_k$ does.  If $r_i > r_j$ holds in the \nth{1} case, and
  $r_j > r_k$ in the \nth{2}, then $min(r_i) < min(r_j) = min(r_k)$
  holds, and so $r_i > r_k$ does; similarly, if $r_i > r_j$ holds in
  the \nth{2} case, and $r_j > r_k$ in the \nth{1}.  If both $r_i >
  r_j$ and $r_j > r_k$ hold in the \nth{2} case, then $min(r_i) =
  min(r_j) = min(r_k)$ and $max(r_i) > max(r_j) > max(r_k)$ hold, and
  so $r_i > r_k$ does.
\end{proof}

% A set of units as a set of RIs.

\subsubsection{A set of RIs}
\label{s:set}

A set of units can be viewed as a set of maximal RIs.  For instance,
$A = \{0, 1\}$ can be viewed as $A = \{\RIO{0}{2}\}$ but not as
$\{\RIO{0}{1}, \RIO{1}{2}\}$.  Set $A$ has $\RIO{0}{2}$, i.e.,
$\RIO{0}{2} \in A$.  However, set $A$ does not have $\RIO{0}{1}$,
i.e., $\RIO{0}{1} \notin A$, but instead includes it, i.e.,
$\RIO{0}{1} \subset A$.  Every set has an empty interval $\emptyset$,
i.e., $\emptyset \in A$.

\subsection{Network model}
\label{s:network}

% The graph.

The network is modeled by a weighted, directed multigraph $G = (V,
E)$, where $V = \{v_i\}$ is a set of vertexes, and $E = \{e_i\}$ is a
set of edges.  Vertex $i$ is denoted by $v_i$.  For edge $e_i$,
function $\cost(e_i)$ gives its cost, and function $\AU(e_i)$ gives
its set of available units (i.e., a set of integers), which do not
have to be contiguous.  Function $\source(e_i)$ gives the source
vertex of edge $e_i$, and $\target(e_i)$ the target.  Function
$\I(v_i)$ gives the set of incoming edges for $v_i$.

% The path.

Path $p = (e_i)$ is a sequence of edges $e_i$ where neighboring edges
meet at the same vertex.  For path $p$, its cost is denoted by
$\cost(p)$, and its set of available units as $\AU(p)$.

% The cost of a path, and the operators.

The algebraic structure of the cost should have two operators defined.
First, the relation $<$ operator should be a total ordering.  Second,
the $\oplus$ operator with the identity 0 should calculate the
combined cost based on its two operands.  The cost of a path is the
cost of its edges combined with $\oplus$ as given by (\ref{e:cost}).
Furthermore, we make Assumption \ref{a:cost}.

% How the cost of a path is produced.

\begin{equation}
  \cost(p) = \bigoplus_{i}\cost(e_i)
  \label{e:cost}
\end{equation}

\begin{assumption}
  Appending edge $e$ to path $p$ does not decrease the cost of the
  path, i.e., $\cost(p) \oplus \cost(e) \ge \cost(p)$.
  \label{a:cost}
\end{assumption}

% Example: the Dijkstra algorithm.

For instance, the costs of integer and real numbers with the less-than
($<$) and the addition ($+$) operators meet these requirements, for
which Assumption \ref{a:cost} becomes $\cost(e) \ge 0$, the
requirement of the Dijkstra algorithm.

% The cost as the product of length and the number of units required.

In optical networks, by searching for a path of minimal length, we can
find a path of the minimal product of the path length $d$ and the
number of units $\n(b, d)$ required by modulation constraints for a
demand of bitrate $b$ at distance $d$.  This implication, given by
(\ref{e:length1}), holds because the number of required units does not
decrease with the increasing length, as given by (\ref{e:length2}).
We note that Assumption \ref{a:cost} rules out electronic signal
regeneration that would lower the number of units required by
modulation constraints.  In \cite{10.1364/JOCN.11.000568}, while
searching for a solution for a demand of bitrate $b$, we discarded any
solution incapable of supporting bitrate $b$ at distance $d$.  In this
work, however, the modulation constraints are not considered during
the search, but can be considered after the search by choosing a
capable solution from those found.

\begin{equation}
  d_i \le d_j \implies d_i\n(b, d_i) \le d_j\n(b, d_j)
  \label{e:length1}
\end{equation}

\begin{equation}
  d_i \le d_j \implies \n(b, d_i) \le \n(b, d_j)
  \label{e:length2}
\end{equation}

% The available units of the path.

The set of available units of path $p$ is available on its every edge
$e_i$, as given by (\ref{e:AU}), which models the resource continuity
constraint.  Furthermore, we make Assumption \ref{a:continuity}, which
follows from (\ref{e:AU}) as the intersection produces a subset of its
operands.  Assumption \ref{a:continuity} rules out the resource
conversion, e.g., the wavelength conversion.

\begin{equation}
  \AU(p) = \bigcap_{i}\AU(e_i)
  \label{e:AU}
\end{equation}

\begin{assumption}[Continuity constraint]
  Appending edge $e$ to path $p$ meets the resource continuity
  constraint, i.e., $\AU(p) \cap \AU(e) \subseteq \AU(p)$.
  \label{a:continuity}
\end{assumption}

% The resource path.

In the standard shortest-path problem, a path is the solution because
there are no resources.  When path $p$ offers more than one RI from
its $\AU(p)$, we need to identify the RI to use, and therefore we
introduce a \emph{resource path} as the solution.  We abbreviate a
resource path to an rpath (read as an R-path).  An rpath $q = (p, r)$
is a pair of path $p$ and RI $r$ denoted by $\RI(q)$, where $\RI(q)
\subseteq \AU(p)$.

\subsection{Label}

% Introduce a label in a flash.  We also note on the label notation.

A label is a summary of an rpath.  Because we make Assumption
\ref{a:label}, we can search for rpaths by comparing their labels.  A
vertex can be reached by several rpaths of interest, and so a vertex
can have a set of labels.  We also say that a label is for a vertex.
An rpath is a feasible solution, so a label is feasible too.  The
label is a pair of the cost and the RI of its rpath.  We denote a
label by $l$, its cost by $\cost(l)$, and the RI by $\RI(l)$.  For
instance, label $l = (0, \Omega)$ is of $\cost(l) = 0$, and $\RI(l) =
\Omega$.

% The fundamental assumption.

\begin{assumption}
  We process labels; we are uninterested in a specific rpath.  Any
  property of an rpath left out of a label is ignored.
  \label{a:label}
\end{assumption}

% The solution equivalence.

Definition \ref{d:equivalence} introduces \emph{equivalent} solutions
for a vertex.  Different vertexes can have equal labels, but they
would not represent equivalent solutions.  For the standard Dijkstra
algorithm, paths of equal cost are equivalent: other properties of a
path, e.g., edges taken or vertexes visited, are ignored.  For the
generic Dijkstra algorithm, rpaths of equal labels are equivalent:
other properties of an rpath, e.g., the set of available units of some
edge, are ignored.  Even when two rpaths are of the same path, they
are not equivalent when their RIs differ.

\begin{definition}[Solution equivalence]
  For a given vertex, solutions of equal labels are considered
  equivalent, because they are indifferentiable by their labels.
  \label{d:equivalence}
\end{definition}

% Equivalent solutions yield equivalent solution.

Even though several equivalent solutions can exist for a vertex, by
Assumption \ref{a:equivalent}, we do not allow a vertex to have equal
labels, which limits the search space and simplifies the search.  To
the best of our knowledge, this assumption was first formalized in
\cite{10.1007/978-3-642-48782-8_9}, where a set with no equivalent
solutions was termed \emph{minimal}.  In that sense, in our work, a
set of labels is intrinsically minimal since a set has unique labels
only.

\begin{assumption}
  Equivalent solutions yield equivalent solutions, so using only one
  of them chosen arbitrarily is correct.  Other equivalent solutions
  are assumed expendable.
  \label{a:equivalent}
\end{assumption}

% How to compare labels: two ways.

Labels are compared for two reasons: first, with $\prec$ by relaxation
to determine whether a label should be kept or discarded; second, with
$<$ by sorting to determine which label should be retrieved from the
priority queue first.  Table \ref{t:l_rels} shows relations between
labels $l_i$ and $l_j$ depending on their costs and RIs, where a cell
in its top-left part reports on the $\prec$ relation, and in the
bottom-right on $<$.

\begin{table*}
  \caption{Relations between labels $l_i$ and $l_j$.}
  \label{t:l_rels}
  \centering%
  \begin{tabular}{c|c|c|c|c|c|}
    \cline{2-6}
    & $\RI(l_i) \subset \RI(l_j)$
    & $\RI(l_i) = \RI(l_j)$
    & $\RI(l_i) \supset \RI(l_j)$
    & \multicolumn{2}{c|}{$\RI(l_i) \parallel \RI(l_j)$}\\
    \hline
    \multicolumn{1}{|c|}{$\cost(l_i) < \cost(l_j)$}
    & \cellcolor[gray]{0.9}\diagbox[dir = SW]{$l_i \parallel l_j$}{$l_i < l_j$}
    & \diagbox[dir = SW]{$l_i \prec l_j$}{$l_i < l_j$}
    & \diagbox[dir = SW]{$l_i \prec l_j$}{$l_i < l_j$}
    & \multicolumn{2}{c|}{\cellcolor[gray]{0.9}
      \diagbox[dir = SW]{$l_i \parallel l_j$}{$l_i < l_j$}}\\
    \hline
    \multicolumn{1}{|c|}{$\cost(l_i) = \cost(l_j)$}
    & \diagbox[dir = SW]{$l_i \succ l_j$}{$l_i > l_j$}
    & \cellcolor[gray]{0}{\color{white}{$l_i = l_j$}}
    & \diagbox[dir = SW]{$l_i \prec l_j$}{$l_i < l_j$}
    & \cellcolor[gray]{0.7}$l_i \parallel l_j$
    & \cellcolor[gray]{0.7}
    \begin{tabular}{@{}c@{}}
      $l_i < l_j \text{ if } \RI(l_i) > \RI(l_j)$\\
      $l_i > l_j \text{ if } \RI(l_i) < \RI(l_j)$
    \end{tabular}\\
    \hline
    \multicolumn{1}{|c|}{$\cost(l_i) > \cost(l_j)$}
    & \diagbox[dir = SW]{$l_i \succ l_j$}{$l_i > l_j$}
    & \diagbox[dir = SW]{$l_i \succ l_j$}{$l_i > l_j$}
    & \cellcolor[gray]{0.9}\diagbox[dir = SW]{$l_i \parallel l_j$}{$l_i > l_j$}
    & \multicolumn{2}{c|}{\cellcolor[gray]{0.9}
      \diagbox[dir = SW]{$l_i \parallel l_j$}{$l_i > l_j$}}\\
    \hline
  \end{tabular}
\end{table*}

% Better, equal, worse, incomparable, efficient, inefficient.

\subsubsection{The $\prec$ relation}

% The better relation.

Label $l_i$ is better than $l_j$ in two cases: \nth{1}, if it offers
at a lower cost a better or equal RI; \nth{2}, if it offers at a lower
or equal cost a better RI.  This relation is given by (\ref{e:l_prec})
and we chose $\prec$ to denote it.\footnote{In
\cite{10.1364/JOCN.11.000568} we introduced this relation with Table 1
and denoted it by $<$.  However, since this relation is partial, here
we traditionally denote it by $\prec$.  The direction (whether to
choose $\prec$ or $\succ$, and call it minimum or maximum efficient)
is debatable; we prefer $\prec$ in order to follow the convention of
\cite{isbn/3540213988}, specifically Table 1.2 there, to denote any
kind of optimality with ``less than''.}  The better or equal relation
is denoted by $\preceq$ and given by (\ref{e:l_preceq}).  A computer
executes $\preceq$ faster than $\prec$.

% The $\prec$ relation for labels.

\begin{equation}
  \begin{split}
    l_i \prec l_j \iff & \cost(l_i) < \cost(l_j) \text{ and}\\
    & \RI(l_i) \supseteq \RI(l_j) \text{ or}\\
    & \cost(l_i) \le \cost(l_j) \text{ and}\\
    & \RI(l_i) \supset \RI(l_j)
  \end{split}
  \label{e:l_prec}
\end{equation}

% The $\preceq$ relation for labels.

\begin{equation}
  \begin{split}
    l_i \preceq l_j \iff & \cost(l_i) \le \cost(l_j) \text{ and}\\
    & \RI(l_i) \supseteq \RI(l_j)
  \end{split}
  \label{e:l_preceq}
\end{equation}

% Label efficiency.

We compare labels with $\prec$ to find \emph{efficient} labels.  A
label is efficient for a given vertex in comparison with other labels
for the \emph{same vertex}.  As given by Definition
\ref{d:efficiency}, an efficient label is the best, i.e., no better
label exists.  If $l_i \prec l_j$, then $l_j$ is inefficient, but it
is undecided whether $l_i$ is efficient.

\begin{definition}[Label efficiency]
  Label $l$ is efficient if, for the same vertex, there does not exist
  label $l^*$ such that $l^* \prec l$.
\label{d:efficiency}
\end{definition}

By Definition \ref{d:expendability}, we can discard label $l'$ if a
better or equal label $l$ exists for the same vertex.  We say that
label $l'$ is expendable in relation with some existing label $l$, or
that labels $l$ and $l'$ are related with $\preceq$, i.e., $l \preceq
l'$.  We note that the difference between efficiency (defined using
$\prec$) and expendability (defined using $\preceq$) is not just the
equality.  The efficiency stipulates that label $l^*$ does \emph{not}
exist, while the expendability that $l$ does exist.

Label $l$ usually belongs to (exists in) some set $L$, and so
(\ref{e:expendability1}) defines relation $\preceq$ for set $L$ and
label $l'$ in the weak (existential) sense: we can discard label $l'$
if there exists in $L$ a better or equal label.  A strong (universal)
sense would be too strong: $l \preceq l'$ would have to hold for every
label $l$.  Furthemore, we extend the expendability relation $\preceq$
to sets, as given by (\ref{e:expendability2}), where $L'$ is
expendable in relation with $L$ if its every label is expendable in
relation with $L$.  Alternative, needlessly strong expendability
definitions for sets are discussed in Appendix \ref{s:needless}.

\begin{definition}[Label expendability]
  A label is expendable if there exists an equal or better label for
  the same vertex.
  \label{d:expendability}
\end{definition}

\begin{equation}
  L \preceq l' \iff \exists l \in L (l \preceq l')
  \label{e:expendability1}
\end{equation}

\begin{equation}
  L \preceq L' \iff \forall l' \in L' (L \preceq l')
  \label{e:expendability2}
\end{equation}

% The minimum of a set of labels.

The minimum of set $L$ of labels is a set of minimal labels, as given
by (\ref{e:min}), i.e., inefficient labels are discarded.  There is no
need to remove equal labels, as a set has only unique labels, agreeing
with Assumption \ref{a:equivalent}.  If $L$ is complete (no other
labels exist), then $\min L$ produces a set of efficient labels.

\begin{equation}
  \min L = \{l \in L : \nexists l' \in L (l' \prec l)\}
\label{e:min}
\end{equation}

% $\preceq$ is a partial ordering, and why exactly.

Labels are $\preceq$-comparable (or comparable in short), if one is
better than or equal to the other, i.e., $\preceq$ or $\succeq$ holds.
However, there can exist lables for which neither $\preceq$ nor
$\succeq$ holds, and so $\preceq$ is a \emph{partial ordering}.
Labels are $\preceq$-incomparable (or incomparable in short), denoted
by $\parallel$, if one is not better than or equal to the other, i.e.,
$\not\preceq$ and $\not\succeq$ hold (the shaded cells in Table
\ref{t:l_rels}), as given by (\ref{e:l_incomparability}).  Relation
$\parallel$ is the symmetric complement of $\preceq$.

\begin{equation}
  l_i \parallel l_j \iff l_i \not\preceq l_j \text{ and } l_i
  \not\succeq l_j
  \label{e:l_incomparability}
\end{equation}

% Label incomparability.

Labels are incomparable for one of two reasons.  First, because their
RIs are incomparable (column 4 of Table \ref{t:l_rels}, regardless of
the cost): e.g., if $l_1 = \lbl{0}{r_1}$, $l_2 = \lbl{1}{r_2}$ and
$l_3 = \lbl{0}{r_3}$, then $l_1 \parallel l_2$ since $r_1 \parallel
r_2$, and $l_2 \parallel l_3$ since $r_2 \parallel r_3$.  Second, if
one label offers a better RI at a higher cost than the other: e.g., if
$l_4 = \lbl{1}{\RIO{0}{3}}$, then $l_1 \parallel l_4$ and $l_3
\parallel l_4$ since $l_4$ offers a better RI at a higher cost than
$l_1$ and $l_3$ (the shaded cell in column 1 of Table \ref{t:l_rels}),
or (conversly) that labels $l_1$ and $l_3$ offer worse RIs at a lower
cost than $l_4$ (the shaded cell in column 3 of Table \ref{t:l_rels}).

% The intransitivity.

The label incomparability is intransitive.  Even though labels are
incomparable for the first reason, e.g., $l_1 \parallel l_2$ and $l_2
\parallel l_3$ (the upper $\parallel$'s in Fig.~\ref{f:ili}), the
transitivity of the incomparability does not always hold, e.g., $l_1
\parallel l_3$ does not hold because $l_1 \succ l_3$ (the lower
$\succ$ in Fig.~\ref{f:ili}).  Next, even though labels are
incomparable for the second reason, e.g., $l_1 \parallel l_4$ and $l_4
\parallel l_3$ (the lower $\parallel$'s in Fig.~\ref{f:ili}), the
transitivity of the incomparability does not always hold, e.g., $l_1
\parallel l_3$ does not hold.  Finally, even though some labels are
incomparable for the first reason, e.g., $l_2 \parallel l_1$, and some
other labels for the second reason, e.g., $l_1 \parallel l_4$, the
transitivity does not hold, e.g., $l_2 \parallel l_4$ does not hold
because $l_2 \succ l_4$ (the upper $\succ$ in Fig.~\ref{f:ili}).
Therefore, $\parallel$ is not an equivalence relation; incomparable
labels are not equivalent.

\begin{figure}
  \centering%
  \begin{tikzpicture}

    \node [node] (l4) {$l_4$};
    \node [node] (l1) at ($(l4) + (210:2 cm)$) {$l_1$};
    \node [node] (l2) at ($(l4) + (90:2 cm)$) {$l_2$};
    \node [node] (l3) at ($(l4) + (330:2 cm)$) {$l_3$};

    \path (l1) edge node [fill = white] {$\parallel$} (l2);
    \path (l2) edge node [fill = white] {$\parallel$} (l3);
    \path (l1) edge node [fill = white] {$\succ$} (l3);

    \path (l4) edge node [fill = white] {$\parallel$} (l1);
    \path (l4) edge node [fill = white] {$\succ$} (l2);
    \path (l4) edge node [fill = white] {$\parallel$} (l3);

  \end{tikzpicture}
  \caption{Intransitivity of label incomparability.}
  \label{f:ili}
\end{figure}

Even though incomparability of labels is not transitive, the
tractability of the stated problem depends on Proposition
\ref{p:savetheday2} related to the incomparability.

\begin{proposition}
  $l_i \parallel l_j \preceq l_k \implies l_i \nsucceq l_j$
  \label{p:savetheday2}
\end{proposition}

\begin{proof}
  Relation $li \parallel l_j$ holds for one of two reasons: first,
  $\RI(l_i) \parallel \RI(l_j)$; second, one of the labels offers a
  better RI at a higher cost than the other.  Furthermore, since $l_j
  \preceq l_k$, then $\cost(l_j) \le \cost(l_k)$ and $\RI(l_j)
  \supseteq \RI(l_k)$.

  In the first case, $\RI(l_i) \parallel \RI(l_j) \supseteq \RI(l_k)
  \implies \RI(l_i) \nsubseteq \RI(l_k)$ by Proposition
  \ref{p:savetheday1}, and that in turn implies $l_i \nsucceq l_j$
  (the first and second column of Table \ref{t:l_rels}).

  In the second case, either $\cost(l_i) < \cost(l_j)$ and $\RI(l_i)
  \subset \RI(l_j)$ hold or $\cost(l_i) > \cost(l_j)$ and $\RI(l_i)
  \supset \RI(l_j)$ hold.  If the former holds, then $\cost(l_i) <
  \cost(l_j) \le \cost(l_k)$ implies $l_i \nsucceq l_k$ (the first row
  of Table \ref{t:l_rels}).  If the latter holds, then $\RI(l_i)
  \supset \RI(l_j) \supseteq \RI(l_k)$ implies $l_i \nsucceq l_k$ (the
  third column of Table \ref{t:l_rels}).
\end{proof}

% On the $prec$ relation.

Relation $\prec$ is a strict partial ordering but not a weak ordering
because its symmetric complement (i.e., $\parallel$) is not an
equivalence relation (since $\parallel$ is intransitive), and so
comparison algorithms cannot use $\prec$.

\subsubsection{The $<$ relation}

% A total ordering <.

For sorting, (\ref{e:l_<}) defines a total ordering $<$ so that either
$<$ or $>$ holds for all nonequal labels.  Ordering $\le$ is given by
(\ref{e:l_<=}).  To compare RIs, (\ref{e:l_<}) uses (\ref{e:r_>}),
while (\ref{e:l_<=}) uses (\ref{e:r_>=}).  Ordering $<$ extends
$\prec$, i.e., $< \supset \prec$, so that $\prec$ implies $<$ (the
three top-right white cells in Table \ref{t:l_rels}) because a better
label should be processed first in accordance with the greedy strategy
(see Section \ref{s:greedy}).

\begin{equation}
  \begin{split}
    l_i < l_j \iff & \cost(l_i) < \cost(l_j) \text{ or}\\
    & \cost(l_i) = \cost(l_j) \text{ and}\\
    & \RI(l_i) > \RI(r_j)
  \end{split}
  \label{e:l_<}
\end{equation}

\begin{equation}
  \begin{split}
    l_i \le l_j \iff & \cost(l_i) < \cost(l_j) \text{ or}\\
    & \cost(l_i) = \cost(l_j) \text{ and}\\
    & \RI(l_i) \ge \RI(r_j)
  \end{split}
  \label{e:l_<=}
\end{equation}

% Lexicographic ordering.  Definition for labels with different costs.

The ordering is lexicographic by the cost first (using $<$), and by
the RI second (using $>$).  We denote it by $<$ since the cost is
compared first using $<$.  For labels of different cost (the top row
of Table \ref{t:l_rels}, which includes the incomparable labels that
are shaded light), we compare costs only: $l_i < l_j$ if $\cost(l_i) <
\cost(l_j)$, which is the first disjunctive clause of (\ref{e:l_<}).

% Definition for labels with equal costs.

For labels of equal cost (the middle row of Table \ref{t:l_rels}), we
break the cost tie with RIs: $l_i < l_j$ if $\cost(l_i) = \cost(l_j)$
and $\RI(l_i) > \RI(l_j)$, which is the second disjunctive clause of
(\ref{e:l_<}).  We compare RIs using $>$ for two reasons.  First, for
the comparable labels (the right white cell of the middle row), $l_i
\prec l_j$ must imply $l_i < l_j$, and it does if $\RI(l_i) >
\RI(l_j)$: for labels of equal cost, $l_i \prec l_j$ holds if
$\RI(l_i) \supset \RI(l_j)$, which in turn implies $\RI(l_i) >
\RI(l_j)$.  Second, for the incomparable labels (shaded dark in Table
\ref{t:l_rels}), the choice is arbitrary: we could flip the relation
for RIs (i.e., compare RIs with $<$ instead of $>$) and $<$ for labels
would still be transitive, but that choice would complicate
(\ref{e:l_<}), which we prefer to keep simple.

% Explain the reverse.

The converse ($>$) for labels holds for the following of Table
\ref{t:l_rels}: the bottom row, the left white cell of the middle row,
and the cell shaded dark (the middle row) if $\RI(l_i) < \RI(l_j)$.

% The ordering is transitive.

\begin{lemma}
  Relation $<$ for labels is transitive.
  \label{l:l_trans}
\end{lemma}

\begin{proof}
  Analogous to the proof of Lemma \ref{l:r_trans}, provided the
  orderings for the values compared are transitive: ordering $<$ for
  the cost is transitive by assumption (Section \ref{s:network}), and
  ordering $>$ for the RI is transitive by Lemma \ref{l:r_trans}.
\end{proof}

The label relations are summerized in Fig.~\ref{f:relations}.  All
label pairs are in the outer set because $\le$ is a total ordering.
The intermediate set (for $\preceq$) contains comparable label pairs.
The inner set (for $\prec$) contains label pairs where one label is
better (and the other worse).

\begin{figure}
  \centering%
  \begin{tikzpicture}

    % Better.
    \node[draw, ellipse, name = better,
      minimum width = 3 cm, minimum height = 0.75 cm] {better};
    \node[fill = white] at (better.north) {$\prec$};

    % Comparable.
    \node[draw, ellipse, name = comparable,
      minimum width = 5 cm, minimum height = 2 cm] {};
    \node[fill = white] at (comparable.north) {$\preceq$};
    \node[above = 0 of comparable.south] {equal};

    % All.
    \node[draw, ellipse, name = all,
      minimum width = 7 cm, minimum height = 3.3 cm] {};
    \node[fill = white] at (all.north) {$\le$};
    \node[above = 0.5 mm of all.south] {incomparable};

  \end{tikzpicture}
  \caption{The label relations.}
  \label{f:relations}
\end{figure}

% How labels get produced.

\subsubsection{How labels get produced}

% Terminology: produced, derived.  How specifically we do it.

Set $L$ of \emph{candidate} labels $l'$ is produced when to a path
described with label $l$ we append edge $e$, which is denoted by $l
\oplus e$.  We say that $l$ and $e$ yield $l'$ or that $l'$ is derived
from $l$ and $e$.  One way of defining $l \oplus e$ is given by
(\ref{e:L}): labels $l'$ have the same cost that depends on the cost
of $l$ and $e$, i.e., $\cost(l') = \cost(l) \oplus \cost(e)$, but
their RIs differ when $\RI(l) \cap \AU(e)$ fans out to a set of RIs.
The nonempty RI of $l'$ is in the intersection of $\RI(l)$ and
$\AU(e)$ to meet the continuity constraint.  Labels in $L$ are
incomparable, and are for the same vertex.

\begin{equation}
  \begin{split}
    L = l \oplus e = \{l' : & \cost(l') = \cost(l) \oplus \cost(e)
    \text{ and}\\
    & \RI(l') \in \RI(l) \cap \AU(e) \text{ and}\\
    & \RI(l') \ne \emptyset \}
  \end{split}
  \label{e:L}
\end{equation}

% The edge is special.

Function $\edge(l)$ gives the edge of label $l$, i.e., the edge that
was appended last to produce the label.  From the mathematical,
minimalistic perspective, the edge is not part of a label and is not
taken into account when comparing labels.  In the standard Dijkstra
too, the preceding vertex is not part of a label but is easily
associated with one because a vertex can have a single label and a
single predecessor; the preceding vertex even does not have to be
stored in a separate predecessor vector, and instead could be found
based on the label, its preceding labels and the graph.  Likewise, in
the generic Dijkstra, the edge of a label could be found based on the
label, its preceding labels and the graph\footnote{In our software
implementation, we associate an edge with a label so that tracing back
a path is easier.  With a label, we could also associate its preceding
label for even easier implementation and faster back tracing, but that
may deteriorate the search performance.}.

% Expendability and the following assumption.

Definition \ref{d:expendability} declares label $l'$ expendable in
relation to label $l$ that we already have in two cases.  First, when
$l = l'$ because we are uninterested in another equal label $l'$ as,
by Assumption \ref{a:equivalent}, we do not differentiate between
equivalent paths.  Second, when $l \prec l'$, which is substantiated
if an inefficient label is unable to yield a better or an incomparable
one for us to take an interest, as stated by Assumption
\ref{a:inefficient}.

\begin{assumption}
  An ineffcient label yields equal or inefficient labels: $l \prec l'
  \implies l \oplus e \preceq l' \oplus e$, for any edge $e$.
  \label{a:inefficient}
\end{assumption}

% There could be other ways to derive labels.

How specifically the set $L$ of candidate labels is produced depends
on the communication network being modeled.  While (\ref{e:L}) defines
the resource interval $\RI(l')$ to be maximal (note the use of $\in$,
as defined in Section \ref{s:set}), some other definitions of $\oplus$
could allow for nonmaximal or even overlapping RIs, provided the
produced labels are incomparable.  If $L$ had comparable labels, then
at least one would be expendable.

\section{Generic greedy approach}
\label{s:greedy}

% The greedy approach in shortest path problems, and the terminology.

The standard and generic Dijkstra algorithms are incarnations of the
\emph{greedy algorithm}, because they use the \emph{greedy approach}.
The greedy algorithm and the greedy approach are complementary
concepts of broad utility \cite{10.1007/978-3-642-58191-5}, concepts
that we build upon and generalize.  We derive the necessary condition
of the generic greedy approach, and the sufficient condition for the
stated problem.

% Terminology: permanent, tentative.

The greedy approach works with labels by Assumption \ref{a:label}.  In
shortest path problems, labels are commonly called \emph{permanent} or
\emph{tentative} (also known as temporary).  A permanent label is
optimal, and is selected from among tentative labels.  A tentative
label is admitted from among candidate labels.  For vertex $i$, the
set of permanent labels is denoted by $P_i$, and the set of tentative
labels by $T_i$.  We collectively refer to the permanent and tentative
labels as the \emph{known} labels.  All known labels are kept in $P =
\{P_i\}$ and $T = \{T_i\}$.  The search starts with an initial
tentative label, and as the search progresses more labels become
known.

% The first way of comparing labels: for the same vertex.  We start
% with this way, because we already know it.  We need it to state that
% a candidate label becomes tentative and then permanent, unless it is
% found expendable.

To decide whether a label is expendable, the greedy approach compares
labels for the \emph{same vertex}.  A candidate label becomes
tentative, unless it gets discarded because an equal (by Assumption
\ref{a:equivalent}) or better (by Assumption \ref{a:inefficient})
label is already known for the same vertex.  A tentative label becomes
permanent, unless it turns out expendable and gets discarded (by
Assumption \ref{a:inefficient}) when a better candidate label becomes
tentative for the same vertex.  Corollary \ref{c:known_incomparable}
follows from expendability.

\begin{corollary}
  Known labels for the same vertex are incomparable.
  \label{c:known_incomparable}
\end{corollary}

\begin{proof}
  Comparable labels are either better, equal or worse.  By Definition
  \ref{d:expendability}, another equal label or a worse label for the
  same vertex would be expendable, and so no such labels are kept,
  leaving only the incomparable.
\end{proof}

% Compare regardless their vertexes to find the best tentative.

In the search for optimal labels, the greedy approach compares labels
for another purpose: to find the best tentative labels.  Tentative
labels are compared \emph{regardless their vertexes}, i.e., their
vertexes are either the same or different.  This seems like comparing
apples to oranges, but the greedy approach does it to reach for the
lowest-hanging fruit, i.e., to select one of the best tentative labels
no matter what its vertex is.

% The tentative labels that compete are for different vertexes.  This
% paragraph elaborates on how the tentative labels are described
% regardless their vertexes.

Tentative labels are compared so that one of the best labels is
selected first, equal and incomparable labels are selected in
arbitrary order, and worse labels are selected or discarded later.  It
is the comparison for different vertexes that finds the best tentative
labels, and not the comparison for the same vertex that is
inconclusive by Corollary \ref{c:known_incomparable}.  Labels for
different vertexes that compare equal or worse are not expendable
because expendability, and in turn inefficiency and equivalance, are
defined for the labels of the same vertex.  One label from among the
best incomparable tentative labels is arbitrarily selected.

% Idea: describe generalization for tentative labels, from the
% standard Dijkstra to the generic Dijkstra.

% How the greedy approach works: local optimality -> optimality.

The greedy approach declares the selected label (\emph{globally})
optimal, thus making it permanent.  The selected label is
\emph{locally} optimal by Definition \ref{d:loptimality} but is
declared (globally) optimal: a better label for the same vertex must
not exist.  This declaration stands for the known labels by Corollary
\ref{c:known_incomparable}, and for the unknown labels by Assumption
\ref{a:unknown}.

\begin{definition}[Local optimality]
  Label $l$ is locally optimal if and only if $l \in \min T$.
  \label{d:loptimality}
\end{definition}

\begin{assumption}
  An unknown label must not be better than any permanent label for the
  same vertex.
  \label{a:unknown}
\end{assumption}

% More on the assumption.

Assumption \ref{a:unknown} is made not only for the selected label,
but for any label that has been selected, i.e., for any permanent
label.  This assumption is equivalent to the well-known assumption
that the local optimality implies (global) optimality, i.e., if a
label is locally optimal (the selected label) then it is (globally)
optimal.  Assumption \ref{a:unknown} entails Corollary
\ref{c:candidate1}.

% The unknown labels -> the candidate labels.

\begin{corollary}
  A candidate label must not be better than any permanent label for
  the same vertex.
  \label{c:candidate1}
\end{corollary}

\begin{proof}
  From among the unknown labels, the greedy approach considers only
  the candidate labels because they are promissing.  The remaining
  unknown labels that could be derived from discarded labels are not
  considered because they would be expendable (by Assumptions
  \ref{a:equivalent} and \ref{a:inefficient}).  Therefore we only have
  to require a candidate label not be better than any permanent label
  for the same vertex.
\end{proof}

% Forefront of optimality.

The selected label is the \emph{forefront of optimality} for two
reasons: first, there is no permanent label ahead of the forefront by
Proposition \ref{p:forefront}; second, there is no tentative label
behind the forefront by Definition \ref{d:loptimality}.  Regardless
their vertexes.  The forefront advances by making the selected label
permanent.

\begin{proposition}
  No permanent label is worse than selected label $l$: $\forall l_p
  \in P (l_p \nsucc l)$.
  \label{p:forefront}
\end{proposition}

\begin{proof}
  We prove by contradiction: $\neg\forall l_p \in P (l_p \nsucc l)
  \iff \neg\nexists l_p \in P (l_p \succ l) \iff \exists l_p \in P (l
  \prec l_p)$, but that is a contradiction, since $l_p$ was assumed
  efficient, and must have been selected before $l$.
\end{proof}

% The invariant.

% Extend later: In the standard greedy approach, the cost of permanent
% labels is smaller or equal to the cost of tentative labels.  No
% tentative label is of smaller cost than some permanent label.

The greedy approach maintains Invariant \ref{invariant} that keeps the
permanent labels safe: for every permanent label $l_p$ there is no
better tentative label $l_t$, regardless their vertexes.  The
invariant holds for tentative labels of the same vertex, i.e., there
is no better tentative label for the vertex of label $l_p$, because
they are the subset of $T$.  As the forefront advances, Invariant
\ref{invariant} is maintained by Corollary \ref{c:maintain1}.

% The invariant below is equivalent to:
%
% $\forall l_p \in P \; \nexists l_t \in T (l_t \prec l_p)$

\begin{invariant}
  $\forall l_p \in P \; \forall l_t \in T (l_p \nsucc l_t)$
  \label{invariant}
\end{invariant}

\begin{corollary}
  Invariant \ref{invariant} is maintained when the selected label is
  made permanent.
  \label{c:maintain1}
\end{corollary}

\begin{proof}
  Making the selected label permanent moves it from $T$ to $P$.
  Before the label is moved, we assume the invariant $\forall l_p \in
  P \; \forall l_t \in T (l_p \nsucc l_t)$ holds.  After label $l$ is
  moved, the invariant $\forall l_p \in P + \{l\} \; \forall l_t \in T
  - \{l\} (l_p \nsucc l_t)$ holds for two reasons.  First, $\forall
  l_p \in P \; \forall l_t \in T - \{l\} (l_p \nsucc l_t)$ holds
  because we assumed the invariant held before moving the label.
  Second, $\forall l_t \in T - \{l\} (l \nsucc l_t)$ holds because $l$
  is one of the best tentative labels.  In short, this invariant is
  maintained because the selected label cannot be rendered suboptimal
  by some tentative label.
\end{proof}

% We need to make a requirement in relation to the selected label.

The selected label is pivotal because it alone yields candidate
labels.  By Proposition \ref{p:maintain2}, a candidate label should be
ahead of the forefront because no permanent label (for whatever
vertex, including the vertex of the candidate label) is there (ahead
of the forefront) that the candidate label could render suboptimal.
Equivalently, to prevent the candidate label get behind the forefront
of optimality (where the optimal labels are), the candidate label must
not be better than the selected label, which is Condition
\ref{c:necessary}.

\begin{proposition}
  Corollary \ref{c:candidate1} holds for a candidate label that
  maintains Invariant \ref{invariant}.
  \label{p:maintain2}
\end{proposition}

\begin{proof}
  A candidate label becomes tentative, if it is not expendable.  If
  that label maintains Invariant \ref{invariant}, i.e., it is not
  better than any permanent label, i.e., regardless their vertexes,
  then Corollary \ref{c:candidate1} holds because it is not better
  than any permanent label for the same vertex.
\end{proof}

\begin{condition}[Necessary condition]
  A label must not produce a better label.
  \label{c:necessary}
\end{condition}

% The condition is necessary.  Is it sufficient?

Condition \ref{c:necessary} is necessary, but whether it is sufficient
for a candidate label to maintain Invariant \ref{invariant} depends on
the transitivity of its not-better-than relation.  By Proposition
\ref{p:forefront}, the selected label is not better than any permanent
label.  By Condition \ref{c:necessary}, a candidate label is not
better than the selected label.  If the non-better-than relation is
transitive, then the candidate label is not better than any permanent
label, Corollary \ref{c:candidate1} holds, and Condition
\ref{c:necessary} is sufficient.  If the not-better-than relation is
intransitive, then Condition \ref{c:necessary} is not sufficient.

% The standard Dijkstra: the label of minimal cost.

The standard Dijkstra algorithm optimizes in the minimum sense: it
selects a locally cheapest label $l$, i.e., there does not exist a
cheaper tentative label, or equivalently, that other tentative labels
$l'$ are no cheaper, i.e., $\cost(l) \ngtr \cost(l')$.  A cheaper
label is better.  To meet Condition \ref{c:necessary}, a label derived
from the selected label $l$ and edge $e$ must not be cheaper, i.e.,
$\cost(l) \ngtr \cost(l) + \cost(e)$ must hold, which is the standard
Dijkstra assumption: $\cost(e) \ge 0$.  Condition \ref{c:necessary} is
sufficient because $\ngtr$ is transitive.

% A small tweak, a big difference.

% The generic Dijkstra: the efficient label.

The generic Dijkstra algorithm optimizes in the efficiency sense: it
selects a locally efficient label $l$, i.e., there does not exist a
better tentative label, or equivalently, that other tentative labels
$l'$ are no better, i.e., $l \nsucc l'$.  To meet Condition
\ref{c:necessary}, labels $l'$ derived from label $l$ and edge $e$
must not be better, i.e., $l \nsucc l'$, or equivalently, that derived
labels $l'$ must be equal to, worse than or incomparable with label
$l$, i.e., $l \preceq l' \text{ or } l \parallel l'$.

% Relation $\nsucc$ is intransitive.

However, since relation $\nsucc$ is intransitive (as $\parallel$ is
intransitive), Invariant \ref{invariant} might not be maintained: even
though Proposition \ref{p:forefront} holds, and even though Condition
\ref{c:necessary} is met, candidate label $l'$ might still render some
permanent label $l_p$ suboptinal, i.e., $\exists l_p \in P (l' \prec
l_p)$ might hold.

% The sufficient condition for the stated problem.

Therefore, we introduce Condition \ref{c:sufficient} that is
sufficient by Proposition \ref{p:sufficient}, and that holds for the
derived labels of the stated problem by Proposition \ref{p:derived}.

% The sufficient condition.

\begin{condition}[Sufficient condition]
  A label must produce an equal or better label.
  \label{c:sufficient}
\end{condition}

\begin{proposition}
  Invariant \ref{invariant} is maintained by a candidate label that
  meets Condition \ref{c:sufficient}.
  \label{p:sufficient}
\end{proposition}

\begin{proof}
  Given permanent label $l_p$, selected label $l$, and candidate label
  $l'$, the following hold: $l_p \nsucc l$ by Proposition
  \ref{p:forefront}, and $l \preceq l'$ by Condition
  \ref{c:sufficient}.  For Invariant \ref{invariant} to be maintained,
  we need to prove that $l_p \nsucc l'$ holds.

  Relation $l_p \nsucc l$ holds in either of two cases: first, if $l_p
  \preceq l$; second, if $l_p \parallel l$.  In the first case, $l_p
  \preceq l \preceq l' \implies l_p \preceq l'$, because $\preceq$ is
  transitive.  In the second case, $l_p \parallel l \preceq l'
  \implies l_p \nsucceq l'$ by Proposition \ref{p:savetheday2}.
  Therefore in both cases $l_p \nsucc l'$ holds.
\end{proof}

\begin{proposition}
  Relation $l \preceq l'$ holds for $l'$ derived from $l$.
  \label{p:derived}
\end{proposition}

\begin{proof}
  For any $l' \in l \oplus e$ and edge $e$, relation $l \preceq l'$
  holds for two reasons.  First, $\cost(l) \le \cost(l')$ by
  Assumption \ref{a:cost}.  Second, $\RI(l) \supseteq \RI(l')$ by
  Assumption \ref{a:continuity}.  In Table \ref{t:l_rels}, the related
  cells are unshaded.
\end{proof}

\subsection{A shortcoming of the generic Dijkstra}
\label{ss:shortcoming}

In \cite{10.1364/JOCN.11.000568}, the generic Dijkstra algorithm was
proposed to sort the labels in the priority queue in the ascending
order of their cost only, and that is a shortcoming, but only when
appending an edge to a path does not increase the cost of the path.

\subsubsection{A failing example}

Figure \ref{f:fe} shows an example where the efficient path from
vertex $s$ to $t$ goes through vertex $u$, is of cost 1 and the RI of
\RIO{0}{2}.  However, generic Dijkstra, depending on the
implementation, could mistakenly find efficient a path of cost 1 and
the RI of \RIO{0}{1} that has edge $e_1$ only.

When visiting vertex $s$, the algorithm would relax edges $e_1$ and
$e_2$ and put labels $l_\RomanI = \lbl{1}{\RIO{0}{1}}$ and $l_\RomanII
= \lbl{1}{\RIO{0}{2}}$ into the queue.  Next, label $l_\RomanI$ could
be retrieved from the priority queue first, thus erroneously finding
this label (for the path of edge $e_1$ only) efficient, while it is
label $l_\RomanII$ (for the path of edges $e_2$ and $e_3$) that is
efficient.

\begin{figure}
  \centering%
  \begin{tikzpicture}

    \node [node] (s) {$s$};
    \node [node, above right = 1.25 cm and 2.5 cm of s] (t) {$t$};
    \node [node, below right = 1.25 cm and 2.5 cm of s] (u) {$u$};

    \path (s) edge
    node [fill = white]
    {$e_1$, $\eattrs{1}{\RIO{0}{1}}$} (t);

    \path (s) edge
    node [fill = white]
    {$e_2$, $\eattrs{1}{\RIO{0}{2}}$} (u);

    \path (t) edge
    node [fill = white]
    {$e_3$, $\eattrs{0}{\RIO{0}{2}}$} (u);

  \end{tikzpicture}
  \caption{A sample failing example.}
  % A failing example.
  \label{f:fe}
\end{figure}

\subsubsection{A correction of the shortcoming}

% Use the $\le$ operator.

Labels in the priority queue should be sorted using the $\le$
relation, as given by (\ref{e:l_<=}) so that RIs are taken into
account too.  Labels for the same vertex are $<$-comparable, but in
the queue there can be equal labels for different vertexes, and so the
$\le$ relation has to be used.  For any labels $l_i$ and $l_j$ in the
priority queue, label $l_i$ will be retrieved before $l_j$ if $l_i \le
l_j$.  Equal labels can be retrieved in arbitrary order.  At the top
of the queue there is the label that is $\le$-comparable with every
other label in the queue.

% Now the failing example works.

With the correction, the failing example above is solved correctly.
When visiting vertex $s$, the algorithm relaxes edges $e_1$ and $e_2$.
In the queue there are two labels: $l_\RomanI = \lbl{1}{\RIO{0}{1}}$,
and $l_\RomanII = \lbl{1}{\RIO{0}{2}}$.  Even though $l_\RomanII \prec
l_\RomanI$, label $l_\RomanI$ is kept as it is for a different vertex.
Label $l_\RomanII$ is retrieved from the queue first, because
$l_\RomanII < l_\RomanI$, and so vertex $u$ is visited next, and label
$l_\RomanIII = \lbl{1}{\RIO{0}{2}}$ is inserted into the queue.  Label
$l_\RomanIII$ is retrieved next, since $l_\RomanIII < l_\RomanI$, thus
label $l_\RomanIII$ is found efficient for vertex $t$.

\section{Generic principle of optimality}
\label{s:bellman}

% The principle of optimality for the shortest path problem (1958).

The principle of optimality was formulated for the shortest-path
problem in \cite{bellman} by (\ref{e:antique}) where: the shortest
paths are found to target vertex $N$ (from all other vertexes), there
are no parallel edges, $f_i$ is the minimal cost of reaching $v_N$
from $v_i$, and $t_{ij}$ is the cost of the edge from $v_i$ to $v_j$
(of $\infty$ if the edge does not exist).  There, the graph was
undirected, so equivalently the search was from $v_N$ to all other
vertexes.  Over the years, for various shortest-path problems, the
principle has been reformulated.

\begin{equation}
  \begin{split}
    f_i &= \min_{j} \{t_{ij} + f_j\}
    \hspace{1cm} \text{if } i \ne N\\
    f_N &= 0
  \end{split}
  \label{e:antique}
\end{equation}

% The principle reformulated: from s, with \oplus, Dijkstra.

\subsection{Reformulation}

If we search for shortest paths from the source vertex $s$ (to all
other vertexes), and there are parallel edges, the principle of
optimality can be formulated by (\ref{e:standard}) where $\oplus$ can
be any operator provided the equations hold for an optimal solution,
e.g., of the lowest cost.  In this formulation, $f_i$ is the minimal
cost of reaching $v_i$ from $v_s$.  Operator $\oplus$ is usually $+$,
but can also be, e.g., the multiplication operator if we search for,
e.g., a path of the lowest probability of failure (i.e., highest
availability).  The relaxation of the Dijkstra algorithm uses this
formulation with the $+$ operator.

\begin{equation}
  \begin{split}
    f_s &= 0\\
    f_i &= \min_{e \in \I(v_i)} (f_{\source(e)} \oplus \cost(e))
    \hspace{1cm} \text{if } i \ne s
  \end{split}
  \label{e:standard}
\end{equation}

% Two assumptions: a total ordering, discard suboptimal labels.

Formulation (\ref{e:standard}) makes two assumptions.  First, a total
ordering of labels is assumed (e.g., the ordering of real numbers),
because a vertex can have one label only: the minimum operation is
expected to provide a single optimal label (e.g., a finite set of real
numbers has a single minimum).  By Assumption \ref{a:equivalent}, from
among equal labels, the minimum operation chooses one arbitrarily.
Second, suboptimal labels are assumed expendable (and so are discarded
by the minimum operation), because, as given by Assumption
\ref{a:suboptimal}, suboptimal label $f_i'$ yields a suboptimal label
$f_i' \oplus \cost(e)$.

\begin{assumption}
  A suboptimal label yields a suboptimal label, i.e., $f_i < f_i'
  \implies f_i \oplus \cost(e) < f_i' \oplus \cost(e)$, for edge $e$
  leaving $v_i$.
  \label{a:suboptimal}
\end{assumption}

Furthermore, this assumption stipulates the optimal substructure: we
know that an optimal label for the target vertex of edge $e$ must be
derived from an optimal label, because deriving it from a suboptimal
one would violate this assumption.  For instance, this assumption
holds for integer and real numbers, and the $+$ operator, by
Assumption \ref{a:cost}.  A problem that violates Assumption
\ref{a:suboptimal}, and thus the principle of optimality, cannot be
solved with the standard Dijkstra algorithm.

% Reformulation needed for a partial ordering.

If the ordering of labels is partial, then a vertex can have a set of
incomparable labels, and so the principle of optimality has to be
generalized.

% The generic principle of optimality.

\subsection{Generalization}

% Introduction to the generic principle of optimality.

We propose a generalization that we call the \emph{generic principle
of optimality}, and which could also be called the generic dynamic
programming principle.  The generalization describes an acceptable
solution that we call the \emph{efficient-resource-path
tree}\footnote{In \cite{10.1109/NOMS56928.2023.10154322} we called it
the efficient-path tree.}, a parallel to the shortest-path tree
described by the principle of optimality.

% A minimal and complete set of efficient labels.

In the efficient-resource-path tree, each vertex has a minimal and
complete set of efficient labels.  Minimal, because from among the
equal labels we keep one by Assumption \ref{a:equivalent}.  Complete,
because we produce all efficient labels, i.e., there does not exist an
efficient label that we could add to the set.

% Description of the equation.

The generalization is given by (\ref{e:generic}), where $C_e$ is the
set of candidate labels we get by traversing edge $e$.  A union of
candidate labels for all edges that lead to $v_i$ form the set of
candidate labels for $v_i$.  Then from the union we take a minimum to
get the set of efficient labels $P_i$ for $v_i$.

% The generic Bellman equations.

\begin{equation}
  \begin{split}
    P_s &= \{(0, \Omega)\}\\
    P_i &= \min \bigcup_{e \in \I(v_i)} C_e
    \hspace{1cm} \text{if } i \ne s
  \end{split}
  \label{e:generic}
\end{equation}

% Labels of edge e.

As given by (\ref{e:candidate1}), candidate labels $C_e$ are derived
from every label $l \in P_{\source(e)}$ and $e$.  $P_{\source(e)}$ is
the set of efficient labels of the source vertex of $e$.

\begin{equation}
  C_e = P_{\source(e)} \oplus e = \bigcup_{l \in P_{\source(e)}} l \oplus e
  \label{e:candidate1}
\end{equation}

% What the generic principle of optimality assumes.

Formulation (\ref{e:generic}) makes two assumptions.  First, a partial
ordering of labels is assumed, because a vertex can have a set of
incomparable labels: the minimum operation is allowed to provide a set
of efficient labels.  By Assumption \ref{a:equivalent}, from among
equal labels, the minimum operation chooses one arbitrarily.  Second,
by Assumption \ref{a:inefficient}, the minimum operation discards
inefficient labels.

% The efficient substructure.

Furthermore, this assumption stipulates the efficient substructure of
a solution: a subproblem has an efficient solution, as given by
Corollary \ref{c:efficient}; in short, an efficient label must be
derived from an efficient label, because deriving it from an
inefficient one would violate this assumption.  A problem that
violates Assumption \ref{a:inefficient}, and thus is out of line with
the generic principle of optimality, cannot be solved with the generic
Dijkstra algorithm.

\begin{corollary}
  An efficient label has the efficient substructure, i.e., has been
  derived from an efficient label.
  \label{c:efficient}
\end{corollary}

\begin{proof}
  % The proof strategy.
  We prove by contradiction: we assume a label is efficient, and
  suppose that its substructure is not to conclude that the assumption
  was false.

  % The proof setting.
  We search for efficient labels from $v_s$, as shown in
  Fig.~\ref{f:substructure}.  To examine the label substructure, we
  assume a resource path to vertex $v'$ is made of a resource path to
  preceding vertex $v$ and the appended edge $e'$.  We assume two
  resource paths lead to $v$: $q_i$ with efficient label $l_i$, and
  $q_j$ with inefficient label $l_j$, and so $l_i \prec l_j$.  The
  label of a resource path to vertex $v'$ is derived from a label of a
  resource path to vertex $v$ and edge $e'$.

  % An efficient path of the inefficient substructure.
  We assume there exists a resource path to vertex $v'$ with efficient
  label $l$, but we suppose an inefficient substructure: label $l$ is
  derived from $l_j$.

  % Assumption contradicted.  Apply the proof recursively.
  By Assumption \ref{a:inefficient}, for every label $l \in l_j \oplus
  e$ there exists $l^* \in l_i \oplus e$ such that $l^* \prec l$, and
  so $l$ cannot be efficient by Definition \ref{d:efficiency}, and
  that contradicts the assumption.  The proof applies to all preceding
  labels (of intermediate vertexes $v$ of a path) as we trace them
  back recursively to finally stop at the initial label.
\end{proof}

\begin{figure}
  \centering%

  \begin{tikzpicture}[node distance = 2.25 cm, ->]
    \node [node] (s) {$s$};
    \node [node, right of = s] (v) {$v$};

    % The middle point between s and v.
    \coordinate (mid) at ($(s)!0.5!(v)$);
    % The label of the upper path.
    \node at ($(mid) + (0, 0.75)$) (qi) {$q_i$};
    % The label of the lower path.
    \node at ($(mid) + (0, -0.75)$) (qj) {$q_j$};

    \node [node, right of = v] (v') {$v'$};

    \draw
    (s) .. controls ($(s) + (-0.25, 1)$) and ($(qi) + (-0.5, -0.25)$) ..
    (qi) .. controls ($(qi) + (0.5, 0)$) and ($(v) + (-0.2, 0.75)$)..
    (v);

    \coordinate (qja) at ($(s) + (0.25, -0.5)$);
    \draw
    (s) .. controls ($(s) + (-0.1, -0.25)$) and ($(qja) + (-0.75, -0.75)$) ..
    (qja) .. controls ($(qja) + (0.25, 0.25)$) and ($(qj) + (-0.5, 0)$) ..
    (qj) .. controls ($(qj) + (1, 0.1)$) and ($(v) + (-1, -0.25)$)..
    (v);

    \draw (v) -- node [circle, fill = white] {$e'$} (v');

  \end{tikzpicture}

  \caption{Efficient substructure.}
  \label{f:substructure}
\end{figure}

\subsection{Constriction}

% We don't want labels with an empty RI.

Set $C_e$ does not include the labels with an empty RI because they
are excluded by (\ref{e:L}) -- that is a constriction, similar to the
standard Dijkstra algorithm constriction: some solutions are deemed
infeasible and are discarded.  While a label with an empty RI would
satisfy the generic principle of optimality, it would not describe a
resource path that meets the resource continuity and contiguity
constraints.  If empty RIs were allowed, the efficient-resource-path
tree would include the shortest-path tree.

% --------------------------------------------------------------------

\section{Problem statement}

Prove the correctness of the generic Dijkstra algorithm that solves
the following problem.

% What we are given.

\noindent{}Given:

\begin{itemize}

\item weighted, directed multigraph $G = (V, E)$, where $V = \{v_i\}$
  is a set of vertexes, and $E = \{e_i\}$ is a set of edges,

\item cost function $\cost(e_i)$, which gives the cost of edge $e_i$,

\item available units function $\AU(e_i)$, which gives the available
  units of edge $e_i$,

\item Condition \ref{c:sufficient} is met,

\item an efficient label meets the generic principle of optimality,

\item the set of all units $\Omega$ on every edge,

\item the source vertex $s$.

\end{itemize}

% What we are looking for.

\noindent{}Find:

\begin{itemize}

\item an efficient-resource-path tree rooted at $s$ that meets the
  resource continuity and contiguity constraints.

\end{itemize}

\subsection{The algorithm}

% The algorithms.

Algorithm \ref{a:algorithm} shows the generic Dijkstra algorithm, and
Algorithm \ref{a:relax} shows the relaxation procedure.

\begin{algorithm}
  \caption{Generic Dijkstra\\
    In: graph $G$, source vertex $s$\\
    Out: an efficient-resource-path tree\\
    \emph{Here we concentrate on permanent labels.}}
  \label{a:algorithm}
  \begin{algorithmic}
    \STATE // The initial label with 0 as the identity of $\oplus$.
    \STATE $T_s = \{\lbl{0}{\Omega}\}$
    \WHILE{$T$ is not empty}
    \STATE $l = \text{pop}(T)$
    \STATE $v = \target(\edge(l))$
    \STATE // Add $l$ to the set of permanent labels for vertex $v$.
    \STATE $P_v = P_v \cup \{l\}$
    \FORALL{out edge $e$ of $v$ in $G$}
    \STATE $\myrelax(l, e)$
    \ENDFOR
    \ENDWHILE
    \RETURN $P$
  \end{algorithmic}
\end{algorithm}

\begin{algorithm}[t]
  \caption{$\myrelax$\\
    In: label $l$, edge $e$\\
    \emph{Here we concentrate on tentative labels.}}
  \label{a:relax}
    \begin{algorithmic}
      \STATE $c' = \cost(l) \oplus \cost(e)$
      \STATE $v' = \target(e)$
      \FORALL{RI $I' \ne \emptyset$ in $\RI(l) \cap \AU(e)$}
      \STATE $l' = (c', I')$
      \STATE // Is label $l'$ expendable?
      \IF{$P_{v'} \not\preceq l'$ and $T_{v'} \not\preceq l'$}
      \STATE // From $T_{v'}$, discard labels $l_{v'}$ worse than $l'$,
      \STATE // leaving only labels incomparable with $l'$.
      \STATE $T_{v'} = T_{v'} - \{l_{v'} \in T_{v'}: l' \preceq l_{v'}\}$
      \STATE // Add $l'$ to the set of tentative labels for vertex $v'$.
      \STATE $T_{v'} = T_{v'} \cup \{l'\}$
      \ENDIF
      \ENDFOR
    \end{algorithmic}
\end{algorithm}

\section{Correctness}
\label{s:correctness}

We propose an inductive proof for the generic Dijkstra algorithm that
also holds for the standard Dijkstra algorithm since a total ordering
(required by the Dijkstra algorithm) is an extension of a partial
ordering (required by the generic Dijkstra algorithm).  A total
ordering is more restrictive as it does not allow for incomparability,
and that simplifies searching.

\subsection{Intuition}

% Correctness for two reasons.

Generic Dijkstra algorithm is correct for two reasons.  First, from
among the tentative labels (that are a link away from efficient
labels), the priority queue provides at the top an efficient label
$l$, because labels are sorted with $\le$: the relation $l \le l'$
holds for any other label $l'$ in the queue, and so no label $l'$ is
better than $l$, i.e., $l \le l' \iff l \preceq l' \text{ or } l
\parallel l' \iff l' \nprec l$.  Second, relaxation replenishes the
priority queue with labels $l'$ derived from $l$, and so, by
Proposition \ref{p:derived}, they cannot be better than $l$.

% Relaxation improves efficiency.

Relaxation is not necessary for the algorithm correctness.  Instead of
relaxation, it would suffice to ensure a candidate label $l'$ for
vertex $v'$ does not represent a loop (a path that returns to one of
its vertexes) before inserting it into the queue, and is not worse
than or equal to any permanent label for vertex $v'$ before making it
permanent.  Such procedure, however, would be inefficient: a candidate
label that is worse than any known label would be inserted, retrieved,
and inevitably discarded.

% How relaxation improves efficiency.  Why $\preceq$, not $\prec$.

Relaxation improves efficiency.  By two steps.  The first inserts into
the queue a candidate label $l'$ only if it is promising at the time
of relaxation, i.e., incomparable with the known labels for vertex
$v'$.  The second discards the tentative labels for vertex $v'$ which
turn out to be worse than the candidate label $l'$ to be inserted into
the priority queue.  Even though we look for worse labels, we use
$\preceq$ because it executes on a computer faster than $\prec$, and
because the equality would never hold, as checked by the first step.

\subsection{Proof}

\begin{theorem}
  The algorithm terminates with a complete set of efficient labels.
\end{theorem}

\begin{proof}

  We prove by induction.  The induction step corresponds to an
  iteration of the main loop.  The induction hypotheses are:

  \begin{enumerate}

  \item $P$ has efficient labels derived from efficient labels,

  \item $T$ has incomparable labels derived from efficient labels.

  \end{enumerate}

  \textbf{Basis.} $|P| = 1$.  In the first iteration, the initial
  label $\lbl{0}{\Omega}$ for vertex $s$ is retrieved from $T_s$ and
  added to $P_s$.  This label is the root of the
  efficient-resource-path tree, and is the only label that has not
  been derived.

  The initial label is efficient because no better label could exist.
  Such a better label would describe a path from $s$ to $s$ either of
  cost lower than 0 (violating Assumption \ref{a:cost}), or of units
  that properly include $\Omega$ (violating Assumption
  \ref{a:continuity}).

  \textbf{Inductive step.} In an inductive step, efficient label $l$
  from the top of $T$ is moved to $P$.  Label $l$ is for vertex $v$.

  % Why the label is efficient.  We use $l'$ for any other label,
  % because it is a derived label that can exists.  We could use a
  % different symbol: $\bar{l}$ if we were assuming that such a better
  % label existed, only to show a contradiction.
  
  To prove that label $l$ is efficient, we show that any other label
  $l'$ for vertex $v$ could not be better, i.e., $l' \prec l$ cannot
  hold.  Label $l'$ must be unknown because otherwise $l$ would not
  be in $T$ by the second hypothesis.

  Label $l'$ must be derived from a known label $l_u$ for some
  preceding vertex $u$.  Label $l_u$ must be tentative because
  otherwise (i.e., $l_u$ is permanent) $l'$ would be known (tentative)
  by the induction hypotheses.  The algorithm would never produce such
  $l'$ derived from tentative $l_u$, but that would prove the
  algorithm incorrect if $l'$ turns out to be efficient.

  Relation $l_u \le l'$ holds by Proposition \ref{p:derived}.
  Relation $l \le l_u$ holds, since it was label $l$ that was
  retrieved from the top of the priority queue.  Relation $l \le l'$
  holds, since $l \le l_u \le l'$, and since $\le$ is transitive by
  Lemma \ref{l:l_trans}.  Therefore $l' \prec l$ cannot hold (because
  $l' < l$ does not, see Fig.~\ref{f:relations}), and so $l$ is
  efficient.
  
  % Hypotheses preserved.
  
  The efficient label $l$ added to $P$ preserves the first hypothesis.
  Relaxation preserves the second hypothesis using the equations of
  the generic principle of optimality: incomparable labels $l'$
  derived from $l$ are produced, and labels worse than $l'$ are
  discarded.

  % Termination: P is complete.
  
  \textbf{Termination.} The algorithm terminates when the priority
  queue gets empty, and then set $P$ is complete.  The algorithm
  terminates as set $P$ is finite: $P$ is replenished by $T$, and the
  number of labels in $T$ is polynomially bounded (see Section
  \ref{s:tractability}).

\end{proof}
  
\section{Tractability}
\label{s:tractability}

We argue the problem is tractable because the size of the search space
is polynomially bounded.  The worst-case number $L$ of labels to
process, which is the number of incomparable labels for all vertexes,
is $L = S|V|$, where $S$ is the number of incomparable labels a vertex
can have.

% The number of incomparable labels.

Number $S$ depends on $U$ only, and is given by (\ref{e:S}).  The
maximal set of incomparable labels has labels whose cost increases
along with the size of their RIs (i.e., the RIs of labels of higher
cost properly include the RIs of labels of lower cost).  The set has
$U$ subsets where labels in each subset can have equal cost: the first
has $U$ labels with RIs of a single unit and the lowest cost; the
second has $U - 1$ labels with RIs of two units and a higher cost,
\ldots; and the last has a single label with the RI of $U$ units and
the highest cost.  The set therefore has $1 + 2 + \ldots + U = (U +
1)U/2$ incomparable labels.

% The conclusion.

Therefore, $O(L) = O(U^2|V|)$, a polynomial bound.

\begin{equation}
  S = \frac{(U + 1)U}{2}
  \label{e:S}
\end{equation}

% An example.

Fig.~\ref{f:labels} draws in black the maximal set of incomparable
labels for three units.  There are also drawn dotted two labels that
are inefficient.  There are three subsets: the first with three labels
of one unit and costs 1, 2 and 3; the second with two labels of two
units and costs 5 and 6; and the third with one label of three units
and cost 8.  We could have made the labels in a subset to have equal
cost (e.g., the second subset with both labels of cost 5), but they
would be harder to plot in the figure.

\begin{figure}
  \centering
  \begin{tikzpicture}
    \begin{axis}[height = 4 cm, width = 6 cm, xlabel = cost,
        ylabel = units]

      \addplot[mark = *] coordinates {
        (1, 1)

        (2, 2)

        (3, 3)

        (5, 1) (5, 2)

        (6, 2) (6, 3)

        (8, 1) (8, 2) (8, 3)};

      \addplot[mark = *, fill = white, densely dotted] coordinates {
        (4, 2)

        (7, 1) (7, 2)};

    \end{axis}
  \end{tikzpicture}
  \caption{The maximal set of incomparable labels for three units.}
  \label{f:labels}
\end{figure}

\section{Conclusion}

Routing of a single connection is one of the most important tasks of
network operations and management.  We have shown that in networks
with discrete resources under the contiguity and continuity
constraints, that task can be efficiently and exactly performed with
the generic Dijkstra algorithm.  The algorithm is not limited to
optical networks, and can be further adapted, similar to how the
Dijkstra algorithm has been adapted over decades.

The novel ideas of resource inclusion, and solution incomparability,
perhaps could be used in other algorithms too, e.g., the Suurballe
algorithm.  The problems that are currently solved with heuristics,
e.g., machine learning, perhaps could be solved exactly and
efficiently with graph algorithms augmented with the proposed novel
ideas, now proven correct.

As future work, the average and worst-case complexities of time and
memory could be analytically evaluated to allow for comparison with
other algorithms, e.g., the filtered-graphs algorithm.  The generic
Dijkstra algorithm has been demonstrated by simulation to be
efficient, but whether the algorithm is the most efficient in
comparison with other algorithms, perhaps even optimal, deserves
further research.

\section{Acknowledgment}

We dedicate this work to Alexander Stepanov for his decades-long
inspiration, and his contributions to generic programming and C++.

\bibliographystyle{unsrt}
\bibliography{all}

\appendices

\section{Needlessly strong expendability}
\label{s:needless}

We could define expendability (denoted by bold $\boldsymbol\preceq$,
different than $\preceq$ defined by (\ref{e:expendability2})) as given
by (\ref{e:needless1a}), i.e., $L'$ would be expendable in relation
with $L$ if for every label $l' \in L'$, there did not exist $l \in L$
that would be incomparable or worse.

\begin{equation}
  L \boldsymbol\preceq L' \iff \forall l' \in L' \; \nexists l \in L
  (l' \parallel l \text{ or } l' \prec l)
  \label{e:needless1a}
\end{equation}

Similarly, we could define expendability as given by
(\ref{e:needless1b}), i.e., $L'$ would be expendable in relation with
$L$ if for every label $l \in L$, there did not exist $l' \in L'$ that
would be incomparable or better.

\begin{equation}
  L \boldsymbol\preceq L' \iff \forall l \in L \; \nexists l' \in L'
  (l' \parallel l \text{ or } l' \prec l)
  \label{e:needless1b}
\end{equation}

Since $(l' \parallel l \text{ or } l' \prec l) \iff l \not\preceq l'$,
we rewrite (\ref{e:needless1a}) and (\ref{e:needless1b}) as
(\ref{e:needless2a}) and (\ref{e:needless2b}).

\begin{equation}
  L \boldsymbol\preceq L' \iff \forall l' \in L' \; \nexists l \in L
  (l \not\preceq l')
  \label{e:needless2a}
\end{equation}

\begin{equation}
  L \boldsymbol\preceq L' \iff \forall l \in L \; \nexists l' \in L'
  (l \not\preceq l')
  \label{e:needless2b}
\end{equation}

Finally, by the De Morgan's law, (\ref{e:needless2a}) and
(\ref{e:needless2b}) are equivalent to (\ref{e:needless}) which is too
strong as it requires every label in $L$ be equal to or better than
every label in $L'$.

\begin{equation}
  L \boldsymbol\preceq L' \iff \forall l \in L \; \forall l' \in L' (l
  \preceq l')
  \label{e:needless}
\end{equation}

\end{document}